\documentclass[12pt]{article}
\usepackage[a4paper,top=2cm,bottom=2cm,left=2cm,right=2cm]{geometry}

\usepackage{amsmath}
\usepackage{amsthm}
\usepackage{amssymb}
\usepackage{braket}
\usepackage{booktabs}
\usepackage{algorithm}
\usepackage{graphicx}
\usepackage{braket}
\usepackage{enumerate}
\usepackage[title]{appendix}
\usepackage[colorlinks=true, allcolors=blue]{hyperref}
\newtheorem{lemma}{Lemma}
\newtheorem{theorem}{Theorem}

\newtheorem{corollary}{Corollary}
\newtheorem{condition}{Condition}
\newtheorem{remark}{Remark}
\theoremstyle{definition}
\newtheorem{definition}{Definition}
\usepackage{algorithm}
\usepackage[numbers]{natbib}

\title{Derandomization of quantum algorithm  for   triangle finding }

\author{ Guanzhong Li\thanks{Email: ligzh9@mail2.sysu.edu.cn}, Lvzhou Li\thanks{Email: lilvzh@mail.sysu.edu.cn (corresponding author)}\\
\small{{\it Institute of Quantum Computing and Software,}}\\
\small{{\it School of Computer Science and Engineering,}}\\
\small {{\it  Sun Yat-sen University, Guangzhou 510006, China}}}
\date{\today }

\begin{document}

\maketitle

\begin{abstract} Derandomization is the process of taking a randomized algorithm and turning it into a deterministic algorithm, which has attracted great attention in classical computing. In quantum computing, it is challenging and intriguing to derandomize quantum algorithms, due to  the inherent randomness of quantum mechanics. The significance of derandomizing quantum algorithms lies not only in theoretically proving
that the success probability can essentially be 1 without sacrificing quantum speedups, but
also in experimentally improving the success rate when the algorithm is implemented on a
real quantum computer.

In this paper, we focus on derandomizing quanmtum algorithms for the triangle sum problem (including the famous triangle finding problem as a special case), which asks to find a triangle in an edge-weighted graph with $n$ vertices, such that its edges sum up to a given weight.
We show that when the graph is promised to contain at most one target triangle, there exists a deterministic quantum algorithm  that either finds the triangle if it exists or outputs ``no triangle'' if none exists. It makes $O(n^{9/7})$ queries to the edge weight matrix oracle, and thus has the same complexity with the state-of-art bounded-error quantum algorithm.  
To achieve this derandomization, we make full use  several techniques:
nested quantum walks with quantum data structure, deterministic quantum search with adjustable parameters,
and dimensional reduction of quantum walk search on Johnson graph.
\end{abstract}

\section{Introduction}
Randomized algorithms play an important role in computer science, as they can be significantly efficient in some choice of computational resource for a lot of basic computational problems, such as time for primality testing~\cite{miller76, Rabin80, SS77}, space for undirected s-t connectivity~\cite{AKLLR79} and circuit depth for perfect matching~\cite{KUW86}.
Since the  polynomial-time deterministic algorithm for primality testing~\cite{AKS04} was proposed in 2004, the study of derandomization, i.e. the question of whether it's possible to come up with efficient deterministic versions of randomized algorithms, has been attracting attention from the academic community, see for instance, Refs.~\cite{Kab02, Russ06, algebraic19, hard21, minimal23}.
Indeed, a lot of exciting works have been dedicated to derandomizing concrete randomized algorithms, including the aforementioned primality testing~\cite{AKS04}, undirected s-t connectivity~\cite{Rein08} and perfect matching~\cite{ST17}.
There are also entire books on derandomization, see for instance, Refs.~\cite{pairwise06, pseudo12, quanti22}.

Because of the inherent randomness of quantum mechanics, most of the existing quantum algorithms are randomized, i.e.,  have a probability of failure~\cite{shor, Grover, Montanaro2016}.  Derandomizing  quantum algorithms seems difficult, with
only  few quantum algorithms having been successfully derandomized (succeed with certainty), such as  deterministic quantum search~\cite{amplitude_amplification,arbi_phase,Long,Roy22,exact} and deterministic quantum algorithms for Simon's problem~\cite{exact_simon} (and its generalization~\cite{GSP}), element distinctness problem~\cite{eedp} and the welded tree problem~\cite{welded_tree23}.
The significance of derandomizing quantum algorithms lies not only in theoretically proving that the success probability  can essentially be 1 without sacrificing quantum speedups, but also in experimentally improving the success rate when the algorithm is implemented on a real quantum computer. Thus, it is intriguing to find more quantum algorithms that allows derandomization. In this paper we will focus on derandomizing quantum algorithms for triangle finding and its generalization, an important and extensively studied problem in quantum computing.

\subsection{Triangle finding and its generalization}
The triangle finding problem has been extensively studied in quantum computing.
It aims to find a triangle in an unknown graph with $n$ vertices, making as few queries as possible to its adjacency matrix given as a black box.
Compared to the classical query complexity of $\Theta(n^2)$, the quantum query complexity of the problem has gradually improved from the trivial $O(n^{3/2})$ using Grover search on triples of the graph's vertices, to the state-of-the-art $O(n^{5/4})$ using extended learning graph~\cite{triangle_extended}.

The first improvement over $O(n^{3/2})$ was given by Buhrman et al.~\cite{Buhrman_element} for sparse graph with $m = o(n^2)$ edges, as they presented a quantum algorithm for the triangle finding problem with query complexity $O(n+\sqrt{nm})$ using amplitude amplification.
Using combinatorial ideas and amplitude amplification, Szegedy~\cite{Szegedy_triangle} (see also~\cite{MagniezSS07}) showed how to solve the problem with query complexity $\widetilde{O}(n^{10/7})$, where $\widetilde{O}(\cdot)$ hides logarithmic factors.
Magniez et al.~\cite{MagniezSS07} then utilized quantum walk search on Johnson graphs, which was originally used to construct an optimal quantum algorithm for the element distinctness problem~\cite{Ambainis07}, to obtain a more efficient algorithm with $\widetilde{O}(n^{13/10})$ queries.
Belovs~\cite{Belovs_learning_graph} introduced the learning graph framework and, as the first application of this framework, used it to improve the quantum query complexity of triangle finding to $O(n^{35/27})$.
Lee et al.~\cite{lee2012improved} then further improved the query complexity to $O(n^{9/7})$, again using learning graph.
Finally, Le Gall~\cite{LeGall_triangle} gave a $\widetilde{O}(n^{5/4})$ quantum algorithm, which utilizes combinatorial structure of the problem, quantum walk search on Johnson graph, and variable time quantum search.
The logarithmic factors in $\widetilde{O}(n^{5/4})$ was later removed using extended learning graphs by Carette et al.~\cite{triangle_extended}.

As can be seen from the above progress, each improvement in the query complexity of triangle finding problem has brought deeper insight into the problem or stimulating new algorithmic technique (See~\cite{Jeffery_review} for a more detailed review).
However, the gap between $O(n^{5/4})$ and $\Omega(n)$ is still open.
At the same time, all the above quantum algorithms are bounded-error, that is, have a probability of failure.

In the above process, Belovs and Rosmanis~\cite{Belovs_power} found that the $O(n^{9/7})$ algorithm by Lee et al.~\cite{lee2012improved} based on non-adaptive learning graph can in fact solve a more general problem --- the triangle sum problem, in which the underlying graph is now edge-weighted and the goal is to find a target triangle whose edges sum to a given value.
They also showed that the algorithm is almost optimal by giving a matching lower bound of $\Omega(n^{9/7}/\sqrt{\log n})$.
Jeffery et al.~\cite{Jeffery_triangle} then proposed a new nested quantum walk with quantum data structures, which is based on the MNRS framework~\cite{MNRS}, and as an application they adapted Lee et al.'s algorithm~\cite{lee2012improved} to a quantum-walk-based algorithm.
However, in order to reduce errors when nesting the bounded-error subroutines, the algorithm makes $\widetilde{O}(n^{9/7})$ queries with additional log factors.

\subsection{Our contribution}
In this paper, we propose a deterministic quantum algorithm for the triangle sum problem (and thus also for the famous triangle finding problem), based on derandomization of the quantum-walk-based algorithm by Jeffery et al.~\cite{Jeffery_triangle}.
Our algorithm has the same $O(n^{9/7})$ query complexity with the state-of-the-art bounded-error quantum algorithm by Lee et al.~\cite{lee2012improved}, but we require an additional promise that the graph has at most one target triangle.
Apart from nested quantum walks with quantum data structures~\cite{Jeffery_triangle}, our algorithm also utilizes deterministic quantum search with adjustable parameters~\cite{exact} (see Section~\ref{subsec:pre_adjustable} and especially Lemma~\ref{lem:beta_fixed}), and a technique to reduce the dimension of invariant subspaces of quantum walk search on Johnson graph (see Section~\ref{subsec:pre_vertex} and especially Lemma~\ref{lem:eedp}), which has also found application in designing a deterministic quantum algorithm for the element distinctness problem~\cite{eedp}.
We think our algorithm has the following significance:
\begin{enumerate}
    \item It's the first deterministic quantum algorithm for the triangle sum (and also triangle finding)  problem making $O(n^{9/7})$ queries, and provides a new example of deranomization of quantum algorithms.
    
    \item It shows the usefulness of the techniques  being utilized, and it's likely that more applications will be found.
\end{enumerate}

\textbf{Formal statement of the triangle sum problem.}
Consider an undirected and weighted simple graph $G$ with $n$ vertices, specified by its edge weight matrix $A \in [M]^{n \times n}$, where $M$ is a positive integer and $[M]:=\{0,1,\cdots,M-1\}$.
The edge weight matrix $A$ can be accessed through a quantum black box (oracle) $O$ whose effect on the computational basis is as follows:
\begin{equation}\label{eq:oracle_triangle_sum}
    O\ket{i,j} \ket{b} \mapsto \ket{i,j}\ket{b\oplus A_{i,j}},
\end{equation}
where $(i,j) \in [n] \times [n]$ encodes an edge of $G$ to be queried, and $A_{i,j} \in [M]$ is the corresponding weight. 
Suppose the value $A_{i,j}$ is stored in $m$ qubits, then $\oplus$ denotes bit-wise XOR and $O^2=I$.

\begin{definition}[the triangle sum problem]\label{prob:triangle_sum}
    Given $d\in [M]$, find three vertices $a,b,c \in [n]$ in the graph $G$ such that
    \begin{equation}\label{eq:def_sum_problem}
        A_{a,b} +A_{b,c} +A_{c,a} = d \mod{M},
    \end{equation}
    making as few queries to the oracle $O$ as possible.
\end{definition}
The triangle sum \textit{promised} problem has an additional promise that if such a target triangle $\triangle abc$ exists in $G$, there is \textit{only} one.

In this paper we obtain the following result.
\begin{theorem}\label{thm:main}
    There is a deterministic quantum algorithm that solves the triangle sum promised problem with certainty and making $O(n^{9/7})$ queries.
\end{theorem}
Specifically, the  algorithm outputs the target triangle if it exists and claims there's none if no such triangle exists.

\begin{corollary}
    There is a deterministic quantum algorithm for triangle finding which makes $O(n^{9/7})$ queries.
\end{corollary}

\begin{proof}
    Triangle finding is a special case of the triangle sum problem, which can be seen by setting $M=4, d=3$ and restricting $A$ to $[2]^{n \times n}$.
    Thus, quantum algorithm solving the triangle sum problem also solves triangle finding.
\end{proof}

Note that the addition in Eq.~\eqref{eq:def_sum_problem} is modulo $M$, which is crucial in proving the $\Omega(n^{9/7}/\sqrt{\log n})$ lower bound of the triangle sum problem~\cite{Belovs_power}.
Intuitively, `addition modulo $M$' makes it impossible for an algorithm to rule out potential triangle, say $\triangle a'b'c'$ in $G$, when the queried edge weights $A_{a',b'}$ and $A_{b',c'}$ already sum up to greater than $d$ or either of them is zero.
A more formal description of this property can be found in~\cite{belovs2012_ksum,Belovs_power} referred to as the orthogonal array condition or the orthogonality property.

\subsection{Paper organization}
The rest of the paper is organized as follows.
In Section~\ref{sec:preliminary} we introduce some important techniques that will be used later: two forms of quantum walk search on Johnson graph dubbed ``edge-walk'' and ``vertex-walk'', a technique to reduce the dimension of the invariant subspace of quantum vertex-walk on Johnson graph, and deterministic quantum search with adjustable parameters.
In Section~\ref{sec:algorithm}, we first sketch the procedure of our algorithm which contains four layers of subroutines in Section~\ref{subsec:alg_outline}, and then show that each layer of subroutine can be derandomized in Section~\ref{subsec:make_layer_zero}, thus obtaining a deterministic quantum algorithm for the triangle sum promised problem.
The details of the quantum walk search on Johnson graph used in three of the four layers, i.e. how to implement the setup, update, checking operations and what is the probability of reaching the target state, are presented in Section~\ref{sec:each_layer}.
We conclude this paper with some related problems in Section~\ref{sec:conclude}.

\section{Preliminaries}\label{sec:preliminary}
In this section, we present some techniques and results that will be used latter in our deterministic quantum algorithm for the triangle sum promised problem.
We will use two forms of quantum walks search on Johnson graph  with subtle differences.
The first one described in Section~\ref{subsec:pre_edge} stems from the nested quantum walk~\cite{Jeffery_triangle}, and we call it ``quantum \textit{edge}-walk on Johnson graph'', since its basis state $\ket{R} \ket{R'}$ can be seen as an edge of a Johnson graph.
The second one described in Section~\ref{subsec:pre_vertex} originates from Ambainis' quantum walks for the element distinctness problem~\cite{Ambainis07}, and we call it ``quantum \textit{vertex}-walk on Johnson graph'', since in its basis state $\ket{R,y}$, $\ket{R}$ can be seen as a vertex of a Johnson graph and $\ket{y}$ as a coin used to update $\ket{R}$.

In Section~\ref{subsec:pre_vertex} we will also introduce a technique (Lemma~\ref{lem:eedp}) to reduce the dimension of the invariant subspace of quantum vertex-walk on Johnson graph, when certain conditions (Condition~\ref{cond:five}) are satisfied.
In order to achieve certainty of success, we will also need deterministic quantum search (with adjustable parameters) as described in Section~\ref{subsec:pre_adjustable}.

\subsection{Quantum edge-walk search on Johnson graph}\label{subsec:pre_edge}
A Johnson graph $J(N,r)$ has $\binom{N}{r}$ vertices. Each vertex is a subset $R\subseteq [N]$ of size $r$, and two vertices $R,R'$ are connected by an edge if and only if $|R\cap R'|=r-1$.
We denote by $V(G)$ the vertex set of $G:=J(N,r)$.

The quantum edge-walk on $G$ has state space $\{ \ket{R}\ket{R'}\ket{D(R)}: R,R'\in V(G) \}$.
The data $D(R)$ associated with $R$ relies on the input of the problem to be solved, and we check if a vertex $R$ is marked based on whether $D(R)$ satisfies certain condition.
The goal of quantum walk search on $G$ is to obtain a marked vertex with constant probability, starting from an initial state $\ket{\psi_0}$ and using update operations that comply with the graph's edges (i.e. to walk on the graph).
Specifically, the quantum edge-walk on $G$ consists of the following three operations.

\textbf{Setup.}
Denote by $S$ the Setup operation that transforms the all zero state to the initial state:
\begin{equation}
    \ket{\psi_0} = S \ket{0},
\end{equation}
where $\ket{\psi_0}$ is an equal superposition of all the edges in $G$:
\begin{equation}
    \ket{\psi_0} = \frac{1}{\sqrt{\binom{N}{r}}} \sum_{R} \ket{R}
\frac{1}{\sqrt{r(N-r)}}\sum_{R \to R'} \ket{R'}
\ket{D(R)},
\end{equation}
where $R \to R'$ means $R'$ is adjacent to $R$ in $G$, or equivalently $|R \cap R'|=r-1$.


\textbf{Update.}
One step of quantum walk, denoted by the update operation $U$, consists of three unitary operators:
\begin{equation}
    U = \mathrm{Data} \cdot \mathrm{Swap} \cdot \mathrm{Coin},
\end{equation}
where `Coin' acting on $\ket{R}\ket{R'}$ is the Grover diffusion of vertices adjacent to $\ket{R}$:
\begin{align}
\mathrm{Coin} &= \sum_{R} \ket{R}\bra{R} \otimes \big( 2\ket{\varphi(R)} \bra{\varphi(R)} - I \big), \\
\ket{\varphi(R)} &= \frac{1}{\sqrt{r(N-r)}}\sum_{R \to R'} \ket{R'},
\end{align}
and
\begin{align}
    & \mathrm{Swap} \ket{R} \ket{R'} = \ket{R'} \ket{R},\\
    & \mathrm{Data} \ket{R'}\ket{R}\ket{D(R)} = \ket{R'}\ket{R}\ket{D(R')}.
\end{align}

\textbf{Checking.} The checking subroutine $C$ adds phase shift $(-1)$ to $\ket{R}\ket{R'}\ket{D(R)}$, if the associated data $\ket{D(R)}$ satisfies certain condition.

The whole process of quantum walk search on $G$ can be formulated in the following equation
\begin{equation}\label{eq:pre_walk_process}
    \ket{\psi_k} = (U^{\sqrt{r}}C)^{k} S \ket{0}.
\end{equation}
The process is similar to Grover's algorithm, as $U^{\sqrt{r}}$ can be seen as a reflection through the initial state $\ket{\psi_0}$ and $C$ as a reflection through the marked states.
It should be noted that the $U^{\sqrt{r}}$ we used in this paper is replaced with phase estimation and selected $(-1)$ phase shift in~\cite{Jeffery_triangle,MNRS}.
Finally, by choosing an appropriate $k$ and measuring $\ket{\psi_k}$ in the first register, we will obtain a marked vertex with constant probability.

\subsection{Quantum vertex-walk search on Johnson graph}\label{subsec:pre_vertex}

The quantum vertex-walk search on Johnson graph $J(N,r)$ has state space $\{ \ket{R,y}\ket{D(R)} : R\subseteq [N], y\in [N]-R \}$.

\textbf{Setup.} The initial state is
\begin{equation}\label{eq:pre_vertex_psi0}
    \ket{\psi_0} = \frac{1}{\sqrt{\binom{N}{r}}} \sum_R \ket{R} \ket{D(R)}
    \frac{1}{\sqrt{N-r}} \sum_{y\in [N]-R} \ket{y}.
\end{equation}

\textbf{Update.} One step of quantum walk $U$ consists of two unitary operators: $U = U_B(\theta_2) U_A(\theta_1)$.
The first operator $U_A(\theta_1)$ acts on $\ket{R,y}$, and can be seen as choosing a random $y\in [N]-R$ to be moved into $R$:
\begin{align}
    U_A(\theta_1) &= \sum_{R} \ket{R}\bra{R} \otimes \big( I-(1-e^{i\theta_1}) \ket{\varphi(R)}\bra{\varphi(R)} \big), \label{eq:pre_U_A_def}\\
    \ket{\varphi(R)} &= \frac{1}{\sqrt{N-r}} \sum_{y\in[N]-R} \ket{y}. \label{eq:pre_A_R_def}
\end{align}
The second operator $U_B(\theta_2)$ can be seen as choosing a random $y'\in R$ being removed from $R$ and at the same time moving $y$ into $R$.
To update $D(R)$ simultaneously, $U_B(\theta_2)$ acts on all registers, but we only define its effect on register $\ket{R,y}$ for simplicity:
\begin{align}
    U_B(\theta_2) &= I- (1-e^{i\theta_2}) \sum_{R+y \subseteq [N]} \ket{\varphi(R+y)} \bra{\varphi(R+y)}, \label{eq:pre_U_B_def}\\
    \ket{\varphi(R+y)} &= \frac{1}{\sqrt{r+1}} \sum_{y'\in R+y} \ket{R+y-y',y'}. \label{eq:pre_B_R_def}
\end{align}
Note that we enhance the original update operation from~\cite{Ambainis07} with general phase shift $e^{i\theta}$ instead of $(-1)$, in order to achieve dimensional reduction of the walk's invariant subspace.

\textbf{Checking.} The checking subroutine $S_\mathcal{M}(\alpha)$ adds relative phase shift $e^{i\alpha}$ to a marked state $\ket{R,y}\ket{D(R)}$, based on whether the associated data $\ket{D(R)}$ satisfies certain condition.
In fact, $S_\mathcal{M}(\alpha)$ can be implemented by $C (I\otimes \mathrm{diag}(1,e^{i\alpha})) C$, if there is a checking subroutine $C$ that flips an auxiliary qubit initialized with $\ket{0}$, when the basis state is marked.
\\

\textbf{5-dimensional invrariant subspace $\mathcal{H}_0$.} 
Suppose the Checking subroutine $C$ satisfies the following condition:

\begin{condition}\label{cond:five}
    There is a special subset $K\subseteq [N]$ with $|K|=2$, such that for a certain $(j_0,l_0) \in \{ (j,0) \}_{j=0}^{2} \cup \{ (j,1) \}_{j=0}^{1}$, the checking subroutine $C$ marks a basis state $\ket{R,y}$ satisfying $|R \cap K| = j_0$ and $|y \cap K| = l_0$.
\end{condition}

Then the quantum vetex-walk search on $J(N,r)$ can be reduced to a $5$-dimensional invariant subspace $\mathcal{H}_0$ spanned by:
\begin{equation}\label{eq:eedp_H0}
    \mathcal{B}_0 := \{ \ket{0,0},\ket{0,1},\ket{1,0},\ket{1,1},\ket{2,0} \},
\end{equation}
where $\ket{j,l}\in\mathcal{B}_0$ is the equal superposition of basis states in $S_{j}^{l} := \{ \ket{R,y} : |R\cap K|=j, |y\cap K|=l \}$.
In basis $\mathcal{B}_0$, the initial state takes the following form:
\begin{equation}
    \ket{\psi_0} = \frac{1}{\sqrt{\binom{N}{r}(N-r)}} \sum_{\ket{j,l}\in \mathcal{B}_0} \sqrt{|S_j^l|}\cdot \ket{j,l},
\end{equation}
and the target state is $\ket{j_0,l_0}$.

Because $S_{j}^{l}$ can also be seen as the set of $\ket{R,y}$ satisfying $|(R+y)\cap K| = j+l$ and $|y\cap K|=l$, the size of $S_{j}^{l}$ can be calculated in two ways.
For $\{ \ket{j,0} \}_{j=0}^{2}$, we have:
\begin{equation}
    \binom{2}{j} \binom{N-2}{r-j} (N-2-(r-j)) = |S_j^0| = \binom{2}{j} \binom{N-2}{r+1-j}(r+1-j),
\end{equation}
and for $\{ \ket{j,1} \}_{j=0}^{1}$, we have:
\begin{equation}
    \binom{2}{j} \binom{N-2}{r-j} (2-j) = |S_j^1| = \binom{2}{j+1} \binom{N-2}{r-j}(j+1).
\end{equation}
This is depicted in Fig.~\ref{fig:pre_subspace} by the two equivalent ways (from left or from right) of calculating the weight of the middle line marked by $\ket{j,l}$ with dashed box.


\begin{figure}
    \centering
    \includegraphics[width=0.6\textwidth]{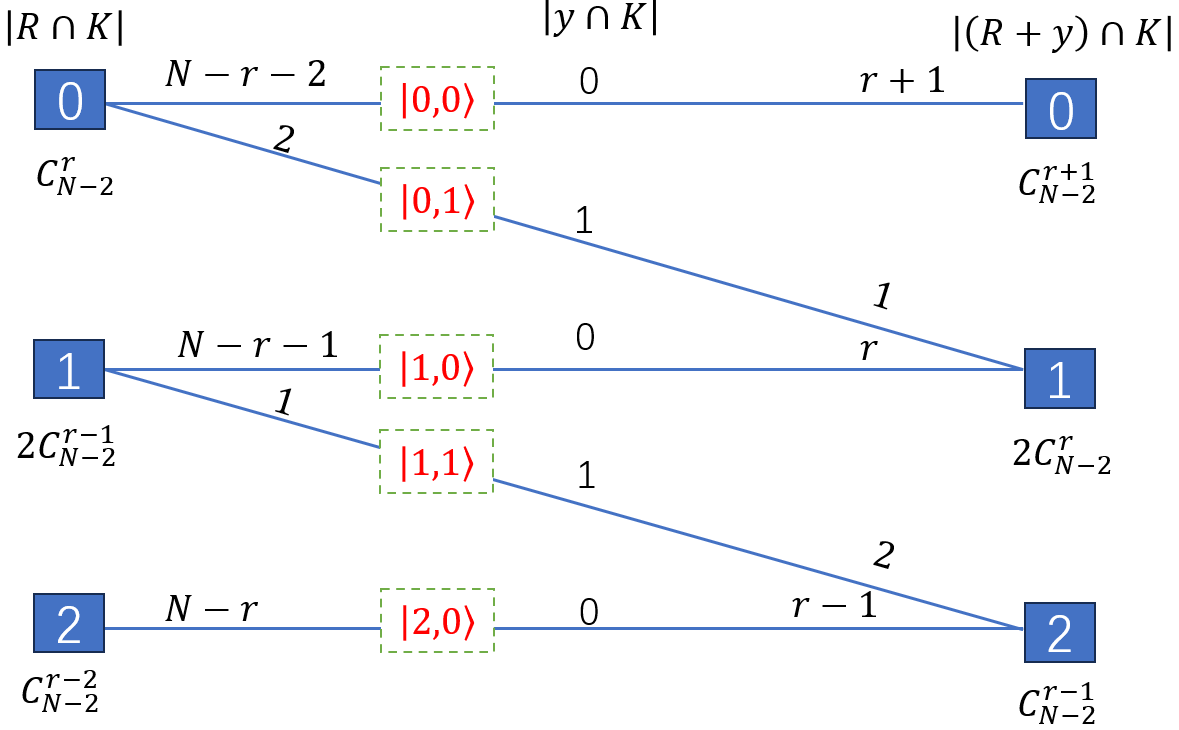}
    \caption{Illustration of the 5 basis states $\ket{j,l} \in \mathcal{B}_0$ of the invariant subspace $\mathcal{H}_0$. The two equivalent ways of calculating $|S_j^l|$ is depicted by the two ways of calculating the weights of the lines. }
    \label{fig:pre_subspace}
\end{figure}

From Fig.~\ref{fig:pre_subspace} and the definition of $U_A(\theta_1)$ and $U_B(\theta_2)$, it can be seen that $U = U_B(\theta_2) U_A(\theta_1)$ takes the following form in $\mathcal{B}_0$:
\begin{align}
    U &= (I-(1-e^{i\theta_2})BB^\dagger) \cdot ( I-(1-e^{i\theta_1})AA^\dagger ),\\
A.^2 &=
\left[
\begin{array}{ccc}
{1-\frac{2}{N-r}} & 0 & 0 \\
{\frac{2}{N-r}} & 0 & 0 \\
0 & {1-\frac{1}{N-r}} & 0 \\
0 & {\frac{1}{N-r}} & 0 \\
0 & 0 & 1
\end{array}
\right],
\quad B.^2 =
\left[
\begin{array}{ccc}
1 & 0 & 0 \\
0 & \frac{1}{r+1} & 0 \\
0 & 1-\frac{1}{r+1} & 0 \\
0 & 0 & \frac{2}{r+1} \\
0 & 0 & 1-\frac{2}{r+1}
\end{array}
\right],
\end{align}
where $A,B$ are non-negative matrices, and $A.^2$ denotes the entry wise square of $A$.
\\

\textbf{Reducing the dimension further to 2.}
We can now state the following Lemma~\ref{lem:eedp} due to Lemmas~1 and 2 in Ref.~\cite{eedp}.

\begin{lemma}\label{lem:eedp}
By setting the parameters $(\theta_1, \theta_2)$ in $U = U_B(\theta_2) U_A(\theta_1)$ and $t \in O(\sqrt{r})$ appropriately,
$U^{t}$ becomes a phase shift of the initial state in $\mathcal{H}_0$:
\begin{equation}
    U^{t} = I-(1-e^{i\beta}) \vert{\psi_0}\rangle \langle{\psi_0}\vert.
\end{equation}
Furthermore, the angle $\beta = t\frac{\theta_1+\theta_2}{2} \ {\rm mod}\  2\pi $ is known and $\beta \approx 1.29\pi$ when $N\to \infty$.
\end{lemma}

Thus the invariant subspace of quantum vertex-walk search on $J(N,r)$, formulated in the following equation:
\begin{equation}
    \ket{\psi_k} = \big( U^t S_\mathcal{M}(\alpha) \big)^k \ket{\psi_0},
\end{equation}
further reduces to $2$-dimensional spanned by $\ket{\psi_0}$ and the target state $\ket{j_0,l_0}$.

\subsection{Deterministic quantum search with adjustable parameters}\label{subsec:pre_adjustable}
Suppose there is a quantum process (unitary operation) $\mathcal{A}$ that transforms the all zero state $\ket{0}$ to $\ket{\psi_0}$, which contains a set of target basis states $\mathcal{M} := \{ \ket{t} \}$.
Denote by $\Pi_\mathcal{M} := \sum_{\ket{t}\in \mathcal{M}} \ket{t}\bra{t} $ the projection onto $\mathrm{span}(\mathcal{M})$.
If the success probability of measuring $\ket{\psi_0}$ that leads to a target state, i.e. $\lambda := \| \Pi_\mathcal{M} \ket{\psi_0} \|^2$ is known beforehand.
Then we can transform $\ket{\psi_0}$ to $\ket{\mathcal{M}} := \frac{1}{\sqrt{|\mathcal{M}|}}\sum_{\ket{t}\in \mathcal{M}}\ket{t}$ with certainty, by iterating the generalized Grover's operation $G(\alpha,\beta)$ for $O(1/\sqrt{\lambda})$ times:
\begin{align}
    G(\alpha,\beta) &:= S_{\psi_0}(\beta) \cdot S_\mathcal{M}(\alpha) \\
    &:= e^{-i\beta \ket{\psi_0} \bra{\psi_0}} \cdot e^{i\alpha\Pi_\mathcal{M}} \\
    &= (\mathcal{A} e^{-i\beta \ket{0} \bra{0}} \mathcal{A}^\dagger) \cdot e^{i\alpha\Pi_\mathcal{M}}.
\end{align}

If the two parameters $\alpha,\beta$ are user-controllable, there are at least three schemes to achieve deterministic quantum search~\cite{amplitude_amplification,arbi_phase,Long}. But the following lemma due to Long~\cite{Long} is most concise.

\begin{lemma}[\cite{Long}]\label{lem:long}
    Suppose parameter $\alpha$ satisfies the equation ``$\sin\frac{\pi}{4k+2} = \sqrt{\lambda} \sin\frac{\alpha}{2}$'', and integer $k > k_\mathrm{opt}$, then
    \begin{equation}
        \left| \bra{\mathcal{M}} [G(\alpha,-\alpha)]^k \ket{\psi_0} \right| =1.
    \end{equation}
    The lower bound of number of iterations $k$ is
    \begin{equation}
        k_\mathrm{opt} = \frac{\pi}{4\arcsin\sqrt{\lambda}} -\frac{1}{2}.
    \end{equation}    
\end{lemma}

If one of the parameters, for example $\beta$, is uncontrollable but fixed and known, we can apply the following lemma from~\cite[Theorem 2]{exact}.

\begin{lemma}\label{lem:beta_fixed}
    Suppose $\beta \in (0,2\pi)$ is arbitrarily given, and $k > k_\mathrm{lower}$, then there always exits a pair of parameters $(\alpha_1, \alpha_2)$ such that
    \begin{equation}
        \left| \bra{\mathcal{M}} [G(\alpha_1,\beta) G(\alpha_2,\beta)]^k \ket{\psi_0} \right| =1.
    \end{equation}
    The lower bound of $k$ is
    \begin{equation}
        k_\mathrm{lower} = \frac{\pi}{ \Big| 4\arcsin(\sqrt{\lambda}\sin\frac{\beta}{2}) \mod [-\frac{\pi}{2},\frac{\pi}{2}] \Big| }\in O(1/\sqrt{\lambda}),
    \end{equation}
    where the notation ``$x \mod [-\frac{\pi}{2},\frac{\pi}{2}]$'' means to add $x$ with an appropriate integer multiples $l$ of $\pi$, such that $x+l\pi\in [-\frac{\pi}{2},\frac{\pi}{2}]$.
\end{lemma}

The explicit equations that $(\alpha_1,\alpha_2)$ need to satisfy in Lemma~\ref{lem:beta_fixed} can be found in~\cite{exact}.

\section{Deterministic quantum algorithm for the triangle sum promised problem}\label{sec:algorithm}
Since our algorithm for the promised problem of triangle sum (Definition~\ref{prob:triangle_sum}) is based on the derandomization of Jeffery et al.'s algorithm~\cite{Jeffery_triangle},
it also has $4$ layers of subroutines as described in Section~\ref{subsec:alg_outline}.
We show the derandomization of each layer in Section~\ref{subsec:make_layer_zero}, and thus obtain a deterministic quantum algorithm and complete the proof of Theorem~\ref{thm:main}.

\subsection{Outline and query complexity}\label{subsec:alg_outline}
Recall that in the triangle sum problem, we are trying to find a triangle whose edges sum up to a given weight $d$ modulo $M$ in an edge-weighted graph $G$ with $n$ vertices, by querying its edge weight matrix $A \in [M]^{n\times n}$ as few times as possible (through oracle $O$ in Eq.~\eqref{eq:oracle_triangle_sum}).

Our algorithm has $4$ layers of subroutines:
layers $1,2,4$ are quantum walk search on Johnson graph, and layer $3$ is a Grover search.
Each layer implements the Check operation of its upper layer.
Intuitively, layer $1$, $2$, and $(3,4)$ as a whole, find one by one, three vertices in the target triangle $\triangle abc := \triangle$.

Denote by $S_i,U_i,C_i$ the setup, update, checking operation in the $i$-th layer,
the query complexity of implementing the respective operations from the oracle $O$ are denoted by $s_i, u_i, c_i$.
We also denote by $\epsilon_i$ the proportion of marked vertices in $V(G_i)$, where $G_i$ is the Johnson graph in the $i$-th layer.
For two subsets $R_1$ and $R_2$ of graph $G$'s vertex set $[n]$, denote by $G_{R_1,R_2}$ the sub-matrix of the edge weight matrix $A$ with rows indexed by $R_1$ and columns indexed by $R_2$.

The 4 layers of subroutines are listed below.
\begin{enumerate}
    \item A quantum edge-walk search on Johnson graph $G_1 = J(n,r_1)$.
    We set $r_1 = n^{4/7}$ so that the total query complexity $c_0$ is $n^{9/7}$ as shown by Eq.~\eqref{eq:overview_total_query} (See \cite{Jeffery_triangle} for why setting $r_1 = n^{4/7}$, and $r_2 = n^{5/7}, m = n^{3/7}$ below).
    A basis state $\ket{R_1} \ket{R_1'} \ket{D(R_1)}$ is marked if $|R_1\cap\triangle|=1$ and $|R_1'\cap\triangle|=1$.
    The query complexity is (with big $O$ omitted)
    \begin{align}
        c_0 &:= s_1 + \frac{1}{\sqrt{\epsilon_1}}( \sqrt{r_1} u_1 +c_1 ) \label{eq:overview_C_0_1st}\\
        &= r_1 r_2 + \sqrt{\frac{n}{r_1}}(\sqrt{r_1}(r_1+r_2) +c_1), \label{eq:overview_C_0_2nd}
    \end{align}
    where the value of $s_1$ and $u_1$ will be explained below.
    The checking operation $C_1$ which checks if $\ket{R_1} \ket{R_1'}$ is marked, is implemented in layer 2 (see also Section~\ref{subsec:make_layer_zero}).
    
    \item A quantum vertex-walk search on Johnson graph $G_2 = J(n_1,r_2)$, where $n_1 = n - r_1$, and we set $r_2 = n^{5/7}$.
    A basis state $\ket{R_1} \ket{R_2,y} \ket{G_{R_1,R_2}}$, where $R_2 \subseteq [n]-R_1$ and $y \in [n] -R_1 -R_2$,
    is marked if $|R_1\cap\triangle| = 1$, $|R_2\cap\triangle| = 1$ and $ y\notin \triangle$.
    To avoid the intolerable Setup cost $s_2 = r_1 r_2$ to construct $\ket{G_{R_1,R_2}}$, nested quantum walks with quantum data structure~\cite{Jeffery_triangle} are used.
    That is, the data $D(R_1)$ associated with $\ket{R_1}$ in layer 1 is the partial initial state (but $\ket{y}$ is not initialized) of layer 2:
    \begin{equation}\label{eq:overview_D_R_1}
        \ket{D(R_1)} = \frac{1}{\sqrt{\binom{n_1}{r_2}}} \sum_{R_2\subseteq[n]-R_1} \ket{R_2} \ket{G_{R_1,R_2}},
    \end{equation}
    where the data $\ket{G_{R_1,R_2}}$ stores $G_{R_1,R_2}$ in a $r_1 \times r_2$ matrix of registers.
    Therefore, to maintain this data structure,
    the setup cost in layer 1 is $s_1 = r_1 r_2$,
    and $u_1 = 2(r_1 + r_2)$ so as to implement $D(R_1) \mapsto D(R_1')$ based on $\ket{R_1'}\ket{R_1}$ (see Lemma~\ref{lem:layer_1_update}).
    Also, $s_2 = 0$, and $u_2 = 2 r_1$ so as to implement $G_{R_1,R_2} \mapsto G_{R_1,R_2'}$.
    Thus, the query complexity of this layer is
    \begin{align}
        c_1 &= s_2 + \frac{1}{\sqrt{\epsilon_2}}( \sqrt{r_2} u_2 +c_2 ) \label{eq:overview_C_1_1st} \\
        &= 0 + \sqrt{\frac{n}{r_2}} (\sqrt{r_2} r_1 + c_2). \label{eq:overview_C_1_2nd}
    \end{align}
    The checking operation $C_2$ which checks if $\ket{R_2,y}$ is marked, is implemented in layer 3.

    \item A Grover search on $[n]- R_1 -R_2 -y$ to find the last vertex $c\in\triangle$, assuming $R_1 \cap \triangle = a$ and $R_2 \cap \triangle = b$.
    The query complexity of this layer is
    \begin{equation}
        c_2 = \sqrt{n} c_3.
    \end{equation}
    The checking operation $C_3$ which checks if $z \in [n]- R_1 -R_2 -y$ is the last vertex $c\in\triangle$, is implemented in layer 4.

    \item A quantum vertex-walk search on the product Johnson graph $G_4 = J(r_1,m) \times J(r_2,m)$ with vertex set $V(G_4) = \{\ket{S_1}\ket{S_2} : S_i \subseteq R_i, |S_i| = m \}$.
    We set $m = (r_1 r_2)^{1/3} = n^{3/7}$.
    Vertices $\ket{S_1}\ket{S_2}$ and $\ket{S_1'}\ket{S_2'}$ are adjacent if $|S_i \cap S_i'|=m-1$ for $i=1,2$.
    The data associated with $\ket{S_i}$ is $G_{S_i,z}$.
    Assuming $R_1 \cap \triangle = a$, $R_2 \cap \triangle = b$ and $z=c$, vertex $\ket{S_1}\ket{S_2}$ is marked if $a\in S_1,b\in S_2$.
    The query complexity of this layer is
    \begin{align}
        c_3 &= s_4 + \frac{1}{\sqrt{\epsilon_4}}( \sqrt{m} u_4 +c_4 ) \\
        &= m + \frac{\sqrt{r_1 r_2}}{m}(\sqrt{m}\cdot O(1) +0).
    \end{align}
    The checking operation $C_4$ flips an auxiliary qubit if there exists $a\in S_1,b\in S_2$ such that $A_{a,b} +A_{b,c} +A_{a,c} = d \,(\mathrm{mod}\, M)$. Since $A_{a,b} \in G_{R_1,R_2}, A_{b,c} \in G_{S_1,c}, A_{a,c} \in G_{S_2,c}$ are already stored in the data structures, $C_4$ requires no oracle query.
\end{enumerate}

We can now calculate the total query complexity as follows:
\begin{align}
    c_3 &= (r_1 r_2)^{1/3} = n^{3/7}, \\
    c_2 &= \sqrt{n} c_3 =n^{6.5/7}, \\
    c_1 &= n^{1/7}(n^{6.5/7} +c_2) =n^{7.5/7}, \\
    c_0 &= n^{9/7} +n^{1.5/7}(n +c_1) =n^{9/7}. \label{eq:overview_total_query}
\end{align}

\subsection{Proof of Theorem~\ref{thm:main}}\label{subsec:make_layer_zero}
The condition that $R_1, R_2, y, z$ are mutually disjoint, and the promise that there's at most one target triangle, enable us to implement the checking operation $C_3$, $C_2$ and then $C_1$ with 100\% success probability successively as shown below.

\textbf{Implementing $C_3$ with 100\% success probability.}
We consider $C_3$ implemented in layer~4 first, whose checking operation $C_4$ does not rely on other layer.
The input state of this layer is $\ket{R_1}\ket{R_2}\ket{z}$, and the quantum walk basis state is $\ket{S_1}\ket{S_2}$ with $S_i \subseteq R_i$.
From the definition of the checking operation $C_4$, we can see that only when the input state $\ket{R_1}\ket{R_2}\ket{z}$ satisfies $|R_1\cap\triangle| = 1$, $|R_2\cap\triangle| = 1$ and $|z\cap\triangle| = 1$, dose $C_4$ not degenerates to the identity operator $I$.
In this special case, the quantum walk process $(U^{t_1} C_4)^{t_2} \ket{\psi_0}$ can be reduced to a $9$-dimensional invariant subspace as shown in Section~\ref{subsec:layer_4}.
Setting $t_1 = \Theta(\sqrt{m})$ and $t_2 = \Theta(\sqrt{r_1 r_2}/m)$,
the success amplitude $p = | \bra{t} (U^{t_1} C_4)^{t_2} \ket{\psi_0}| = \Omega(1)$ by Lemma~\ref{lem:layer_4_success},
where the target state $\ket{t}$ is an equal superposition of $\ket{S_1}\ket{S_2}$ with $a\in S_1, b\in S_2$ (assuming $a\in R_1, b\in R_2$),
and $p$ can be computed exactly beforehand.
Thus combined with Lemma~\ref{lem:long}, we can obtain $\ket{t}$ with certainty.
Denote by $\mathcal{A}_4$ the above process that on input state $\ket{R_1}\ket{R_2}\ket{z}$, implements $\ket{0} \mapsto \ket{t}$.
Then applying $C_4$ once more we can flip an auxiliary qubit with certainty.
Thus $\mathcal{A}_4^\dagger C_4 \mathcal{A}_4$ implements $C_3$ with 100\% success probability, and it acts nontrivially when $|R_1\cap\triangle| = 1$, $|R_2\cap\triangle| = 1$ and $|z\cap\triangle| = 1$.

\textbf{Implementing $C_2$ with 100\% success probability.}
The input state of layer 3 is $\ket{R_1}\ket{R_2,y}$, and the search space is $z\in [n] -R_1 -R_2 -y$.
Thus combined with the effect of $C_3$ shown above, we can see that only when the input state $\ket{R_1}\ket{R_2,y}$ satisfies $|R_1\cap\triangle| = 1$, $|R_2\cap\triangle| = 1$ and $y\notin\triangle$, does $C_3 \neq I$.
In this special case, suppose $a\in R_1, b\in R_2$, then the only target $t\in [n] -R_1 -R_2 -y$ is $c\in\triangle$.
Therefore, using the deterministic quantum search shown by Lemma~\ref{lem:long}, we can obtain $\ket{t}$.
Denote by $\mathcal{A}_3$ the above process that on input state $\ket{R_1}\ket{R_2,y}$, implements $\ket{0} \mapsto \ket{t}$.
Then applying $C_3$ once more we can flip an auxiliary qubit with certainty, and thus $\mathcal{A}_3^\dagger C_3 \mathcal{A}_4$ implements $C_2$ with 100\% success probability, which acts nontrivially when $|R_1\cap\triangle| = 1$, $|R_2\cap\triangle| = 1$ and $y\notin\triangle$.

\textbf{Implementing $C_1$ with 100\% success probability.} Recall that we want $C_1$ to act nontrivially on $\ket{R_1} \ket{R_1'} \ket{D(R_1)}$ when $|R_1\cap\triangle|=1$ and $|R_1'\cap\triangle|=1$.
We will implement $C_1$ from $\bar{C}_1$, which acts nontrivially on $\ket{R_1} \ket{D(R_1)}$ when $|R_1\cap\triangle|=1$.
The implementation of $C_1$ from $\bar{C}_1$ is described in Section~\ref{subsec:layer_1}, which requires $c_1 = 2u_1 + 4\bar{c}_1$ queries.
Since $u_1 = n^{5/7} \ll \bar{c_1} = n^{7.5/7}$, we have $c_1 = \Theta(\bar{c}_1)$.
We now describe how to implement $\bar{C}_1$.
The input state is $\ket{R_1}\ket{D(R_1)}$, and the quantum walk basis state is $\ket{R_2,y}$ with $R_2\subseteq [n]-R_1$, $y\in [n] -R_1 -R_2$.
From the effect of $C_2$ shown above, we can see that only when the input state satisfies $|R_1\cap\triangle| = 1$, does $C_2 \neq I$.
In this special case, $|([n] -R_1) \cap \triangle| = 2$, and thus the quantum walk search process satisfies Condition~\ref{cond:five} with $K = \triangle - R_1$ and $(j_0,l_0)=(1,0)$.
Therefore, the walk can be reduced to a 2-dimensional invariant subspace by Lemma~\ref{lem:eedp}, and then by Lemma~\ref{lem:beta_fixed} we can obtain with certainty the target state $\ket{t}$,
which is an equal superposition of $\ket{R_2,y}$ with $|R_2\cap\triangle| = 1$ and $y\notin\triangle$.
See Section~\ref{subsec:layer_2} for more details.
Denote by $\mathcal{A}_2$ the above process that on input state $\ket{R_1}\ket{D(R_1)}$, implements $\ket{0} \mapsto \ket{t}$.
Then applying $C_2$ once more we can flip an auxiliary qubit with certainty,
and thus $\mathcal{A}_2^\dagger C_2 \mathcal{A}_2$ implements $\bar{C}_1$ with 100\% success probability.

\begin{remark}\label{rem:layer2}
It's worth noting that in the implementation of $\bar{C}_1$ in layer 2, we cannot use Lemma~\ref{lem:long} to obtain with certainty the target state like in the implementation of $C_3$ in layer~4.
This is because in the implementation of the phase shift operator $S_{\psi_0}(\beta) = e^{-i\beta \ket{\psi_0} \bra{\psi_0}}$ in Lemma~\ref{lem:long}, we would otherwise need to query $r_1 \times r_2$ times to construct $\ket{G_{R_1,R_2}}$ in $\ket{\psi_0}$,
which is intolerable since $r_1 \times r_2 = n^{9/7}$ would then be multiplied by $1/\sqrt{\epsilon_1} = n^{1.5/7}$ of layer 1, making the query complexity exceed $n^{9/7}$.
\end{remark}

\textbf{Derandomizing layer 1 and finishing the proof of Theorem~\ref{thm:main}.}
As shown above, $C_1$ acts nontrivially only when basis state $\ket{R_1}\ket{R_1'}\ket{D(R_1)}$ satisfies $|R_1\cap\triangle| = 1$ and $|R_1'\cap\triangle| = 1$.
Therefore, the quantum walk search on $G_1 = J(n,r_1)$ can be reduced to a $10$-dimensional invariant subspace, as shown in Section~\ref{subsec:layer_1}.
Also, by setting $t_1 = \lfloor \frac{\pi}{2} \sqrt{2 r_1} \rceil$ and $t_2 = \lfloor \frac{\pi}{4} \sqrt{\frac{n}{3 r_1}} \rceil$,
the success amplitude $p = \left| \bra{t} (W^{t_1} C_1)^{t_2} \ket{\psi_0} \right| = \Omega(1)$ as shown by Lemma~\ref{lem:layer_1_success},
and $p$ can be computed exactly beforehand.
Thus combined with Lemma~\ref{lem:long}, we can obtain with certainty the target state $\ket{t}$,
which is an equal superposition of $\ket{R_1}$ satisfying $|R_1\cap\triangle| = 1$.

We can now apply the deterministic (quantum walk) search $\mathcal{A}_2, \mathcal{A}_3, \mathcal{A}_4$ of layers 2,3,4 successively and then measure all the registers.
This will lead to $\ket{R_1} \ket{R_2} \ket{z} \ket{S_1} \ket{S_2}$ satisfying $a\in S_1$, $b\in S_2$ and $z=c$.
Thus from the associated data $\ket{G_{R_1,R_2}} \ket{G_{S_1,z}} \ket{G_{S_2,z}}$ we can find the target $\triangle abc$ with certainty.
If there's no $a\in S_1, b\in S_2$ such that $A_{a,b} +A_{a,z} +A_{b,z} = d$, we claim that the graph $G$ does not contain the target triangle.


\section{Details of quantum walk search in different layer}\label{sec:each_layer}
In the following subsections, we will fill in the missing details of the three aforementioned quantum walk search on Johnson graphs:
(i) edge-walk on $G_1 = J(n,r_1)$ of layer 1,
(ii) vertex-walk on $G_2 = J(n_1,r_2)$ of layer 2,
and (iii) vertex-walk on the product Johnson graph $G_4 = J(r_1,m)\times J(r_2,m)$ of layer 4.
We start with the layer 4 whose checking operation does not rely on other layer.

\subsection{Layer 4}\label{subsec:layer_4}
Suppose the input state $\ket{R_1} \ket{R_2} \ket{G_{R_1,R_2}}\ket{z}$ of this layer satisfies $a\in R_1, b\in R_2, z=c$,
the target state of quantum vertex-walk search on $G_4 = J(r_1,m) \times J(r_2,m)$ is an equal superposition of $\ket{S_1}\ket{S_2}$ such that $a\in S_1, b\in S_2$.

\textbf{Setup.}
The initial state is
\begin{equation}
    \ket{\psi_0} = \bigotimes_{i=1}^{2}
    \frac{1}{\sqrt{\binom{r_i}{m}}} \sum_{S_i \subseteq R_i} \ket{S_i} \ket{G_{S_i,z}}
    \frac{1}{\sqrt{r_i-m}}\sum_{z_i\in R_i-S_i} \ket{z_i}.
\end{equation}
We first construct an equal superposition of $\ket{S_i}$ s.t. $S_i \subseteq R_i$ and $|S_i| = m$ based on $\ket{R_i}$,
and then query the oracle for $m$ times to construct the associated data $\ket{G_{S_i,z}}$ based on $\ket{S_i}\ket{z}$,
and finally construct an equal superposition of $\ket{z_i}$ s.t. $z_i \in R_i-S_i$ based on $\ket{R_i}\ket{S_i}$.

\textbf{Update.} A step of quantum walk $U$ consists of $2$ query operations and $2$ diffusion operations:
\begin{equation}
    U = Q \cdot U_B \cdot Q \cdot U_A,
\end{equation}
with working registers:
\begin{equation}
    \ket{z} \bigotimes_{i=1}^{2} \ket{R_i} \ket{S_i}\ket{z_i} \ket{G_{S_i,z}}\ket{G_{z_i,z}}.
\end{equation}
The first diffusion operation $U_A$ acts on registers $\bigotimes_{i=1}^{2} \ket{R_i}\ket{S_i}\ket{z_i}$, and can be seen as choosing a random $z_i \in R_i-S_i$ to be moved into $S_i$:
\begin{align}
U_A &= \sum_{R_1,R_2} \sum_{S_1\subseteq R_1} \sum_{S_2\subseteq R_2} \bigotimes_{i=1}^{2} \ket{R_i,S_i} \bra{R_i,S_i}
\otimes \Big( 2 \bigotimes_{i=1}^{2} \big( \ket{\varphi(R_i,S_i)} \bra{\varphi(S_i,R_i)} \big) - I \Big), \\
\ket{\varphi(S_i,R_i)} &=\frac{1}{\sqrt{r_i-m}}\sum_{z_i\in R_i-S_i} \ket{z_i}.
\end{align}
The sum $\sum_{R_1,R_2}$ is over $R_1 \subseteq [n]$, $R_2 \subseteq [n]-R_1$, and for other $R_2, S_1, S_2$ we define $U_A$ to act trivially on $\ket{z_1} \ket{z_2}$.
The query operation $Q$ calls the oracle $O$ (Eq.~\eqref{eq:oracle_triangle_sum}) on registers $\ket{z_i} \ket{z} \ket{G_{z_i,z}}$ for $i=1,2$.
The second diffusion operation $U_B$ acts on all registers except $\ket{z}$, and can be seen as choosing a random $z_i' \in S_i$ being removed from $S_i$ and at the same time moving $z_i$ into $S_i$, while updating the associated data $\ket{G_{S_i,z}} \ket{G_{z_i,z}}$ simultaneously:
\begin{align}
U_B &= \sum_{R_1,R_2} \ket{R_1,R_2}\bra{R_1,R_2} \otimes C_{R_1,R_2}, \\
C_{R_1,R_2} &= 2\bigotimes_{i=1}^{2} \left( \sum_{S_i + z_i\subseteq R_i}
\sum_{\bar{G}_i \in [M]^{m+1}}
\ket{\varphi_{S_i + z_i}^{\bar{G}_i}} \bra{\varphi_{S_i + z_i}^{\bar{G}_i}} \right) - I, \\
\ket{\varphi_{S_i + z_i}^{\bar{G}_i}} &= \frac{1}{\sqrt{m+1}} \sum_{z_i'\in S_i + z_i} \ket{S_i'} \ket{z_i'} \ket{G_{S_i',z}} \ket{G_{z_i',z}}. \label{eq:layer_4_phi_B}
\end{align}
In $\ket{\varphi_{S_i + z_i}^{\bar{G}_i}}$,
suppose the vertical juxtaposition $[G_{S_i,z};G_{z_i,z}] = \bar{G}_i$,
then for $\ket{S_i'} \ket{z_i'}$ s.t. $S_i' = S_i -z_i' +z_i$,
its associated data $[G_{S_i',z};G_{z_i',z}]$ is a corresponding permutation of $\bar{G}_i$:
we exchange row $z_i$ and $z_i'$ in $\bar{G}_i$,
and then sort the first $m$ rows in the ascending order.
Thus, if $\bar{G}_i$ is a sub-matrix of the edge weight matrix $A$, then $G_{S_i',z}$ is still a sub-matrix of $A$ with rows indexed by $S_i'$ and columns indexed by $z$.

\textbf{Checking.} The checking operation $C_4$ flips an auxiliary qubit if there exists $a\in S_1, b\in S_2$ such that $A_{a,b} +A_{b,z} +A_{a,z} = d$,
based on the data $\ket{G_{S_i,z}}$ constructed in layer 4, and $\ket{G_{R_1,R_2}}$ constructed in layer 1.
Thus no additional query is needed.
Using the phase kick-back effect: $X \ket{-} = -\ket{-}$, where $X$ is the Pauli-X matrix and $\ket{-} := (\ket{0}-\ket{1})/\sqrt{2}$, we can add $(-1)$ phase shift to the marked states.

\textbf{Invariant subspace.} Denote by $\ket{(j_1,j_2)-(k_1,k_2)}$ the equal superposition of states in $\{ \ket{S_1,z_1,S_2,z_2} : |S_i\cap \triangle|=j_i, |(S_i+z_i)\cap\triangle| =k_i \}$.
The $9$ basis states of the quantum walk's invariant subspace $\mathcal{H}_0$ is illustrated in Fig.~\ref{fig:layer_4_subspace}.

\begin{figure}
    \centering
    \includegraphics[width=0.75\textwidth]{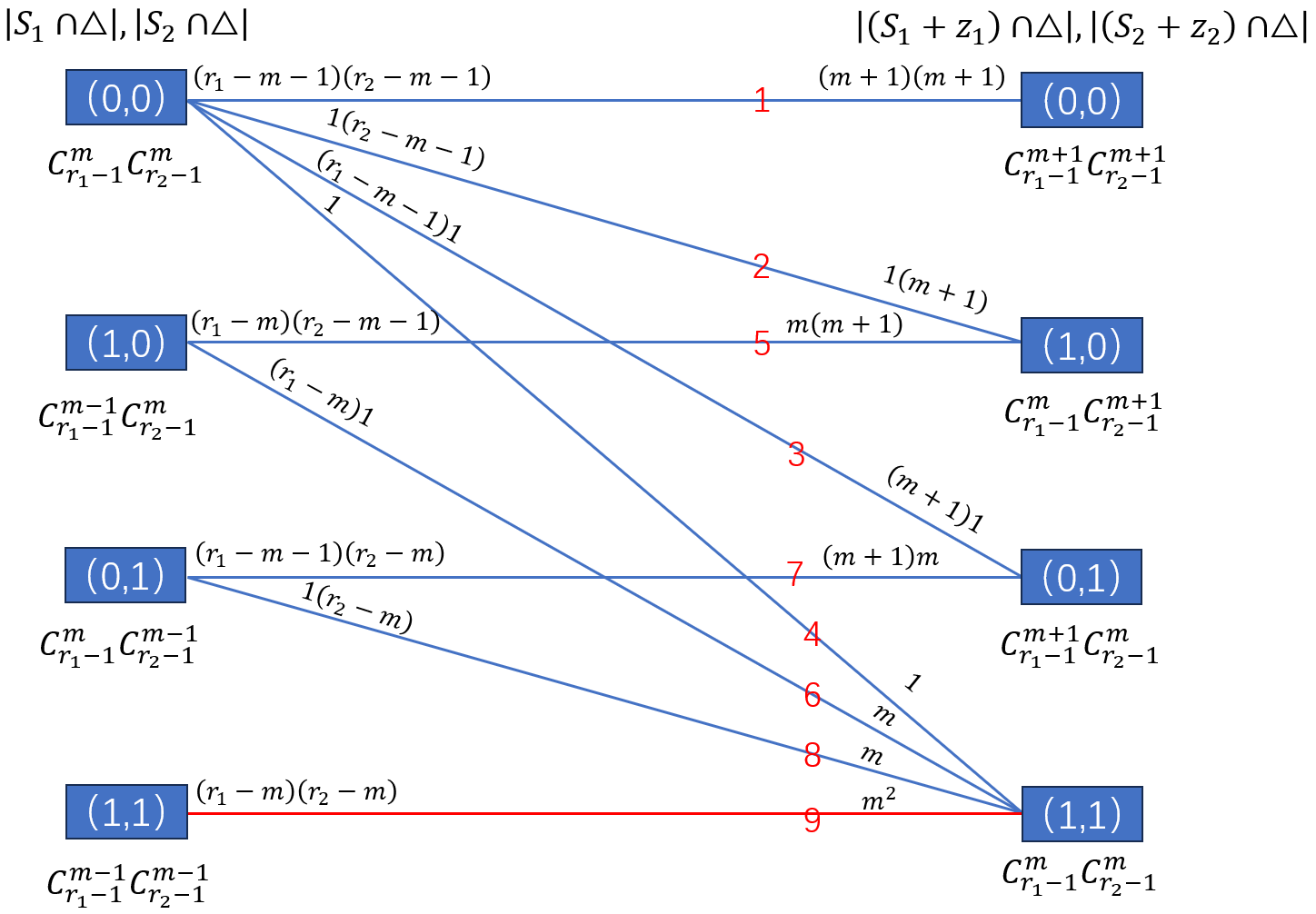}
    \caption{Illustration of the $9$ basis states $\ket{(j_1,j_2)-(k_1,k_2)}$ of the invariant subspace $\mathcal{H}_0$ of quantum walk search on $G_4 = J(r_1,m) \times J(r_2,m)$. }
    \label{fig:layer_4_subspace}
\end{figure}

It can be seen form Fig.~\ref{fig:layer_4_subspace} that a step of quantum walk $U$ takes the following matrix form in $\mathcal{H}_0$.
\begin{equation}\label{eq:W_layer4}
W=(2BB^\dagger-I)(2AA^\dagger-I),
\end{equation}
where non-negative matrices $A$ and $B$ satisfy:
\begin{equation}
A.^2=\begin{bmatrix}
(1-\frac{1}{r_1-m})(1-\frac{1}{r_2-m}) & 0 & 0 & 0 \\
\frac{1}{r_1-m} (1-\frac{1}{r_2-m}) & 0 & 0 & 0 \\
(1-\frac{1}{r_1-m})\frac{1}{r_2-m} & 0 & 0 & 0 \\
\frac{1}{(r_1-m)(r_2-m)} & 0 & 0 & 0 \\
0 & 1-\frac{1}{r_2-m} & 0 & 0 \\
0 & \frac{1}{r_2-m} & 0 & 0 \\
0 & 0 & 1-\frac{1}{r_1-m}  & 0 \\
0 & 0 & \frac{1}{r_1-m}  & 0 \\
0 & 0 & 0 & 1
\end{bmatrix},
B.^2=\begin{bmatrix}
1 & 0 & 0 & 0 \\
0 & \frac{1}{m+1} & 0 & 0 \\
0 & 0 & \frac{1}{m+1} & 0 \\
0 & 0 & 0 & \frac{1}{(m+1)^2} \\
0 & \frac{m}{m+1} & 0 & 0 \\
0 & 0 & 0 & \frac{m}{(m+1)^2} \\
0 & 0 & \frac{m}{m+1} & 0 \\
0 & 0 & 0 & \frac{m}{(m+1)^2} \\
0 & 0 & 0 & \frac{m^2}{(m+1)^2} \\
\end{bmatrix}
\end{equation}
From the number of basis states in $\ket{(j_1,j_2)-(k_1,k_2)}$ (or weights of the $9$ lines in Fig.~\ref{fig:layer_4_subspace}),
it's easy to see the initial state takes the following form in $\mathcal{H}_0$ after simplification:
\begin{equation}
    \ket{\psi_0}.^2 =
    \begin{bmatrix}
        (r_1 -m -1) (r_2 -m -1) \\
        (r_2 -m -1) \\
        (r_1 -m -1) \\
        1 \\
        m (r_2 -m -1) \\
        m \\
        (r_1 -m -1) m \\
        m \\
        m^2
    \end{bmatrix}
    \div (r_1 r_2),
\end{equation}
and the overlap with the target state $\ket{t} = \ket{e_9}$ is $\sqrt{\epsilon_4} = \frac{m}{\sqrt{r_1 r_2}}$.

\begin{lemma}\label{lem:layer_4_success}
    Setting $t_1 = \lfloor \frac{\pi}{2} \sqrt{\frac{m}{2}} \rceil$ and $t_2 = \lfloor \frac{\pi}{4} \frac{\sqrt{r_1 r_2}}{m} \rceil$, the success amplitude $p = \left| \bra{t} (W^{t_1} C_4)^{t_2} \ket{\psi_0} \right| $ satisfies $ p = 1 - O(\frac{m}{r_2} +\frac{m}{r_1} +\frac{1}{m})$.
\end{lemma}

\begin{proof}
    Similar to \cite[Lemma 2]{subset}, it can be shown that $W$ has two eigenvectors $\ket{u_\pm}$ with corresponding eigenvalues $e^{\pm i \varphi}$ such that
    \begin{align}
        \left| \braket{t|u_{\pm}} \right| &= \frac{1}{\sqrt{2}} + O(\frac{1}{m} +\frac{m}{r_1} +\frac{m}{r_2}), \\
        \varphi &= \frac{2\sqrt{2}}{\sqrt{m}}(1+O(\frac{1}{\sqrt{m}} +\frac{m}{r_1} +\frac{m}{r_1})).
    \end{align}
    Similar to \cite[Lemma 3]{subset},
    it can be shown that when $t_1 \varphi \approx \pi$ and $C_4 = 2\ket{t}\bra{t} - I$ adds relative phase shift $(-1)$ to the only target basis state $\ket{t}$ in $\mathcal{H}_0$,
    $W^{t_1}C_4$ has two eigenvectors $\ket{\theta_\pm}$ with corresponding eigenvalues $e^{\pm i\theta}$ such that
    \begin{align}
        \braket{t|\theta_{\pm}} &= \frac{1}{\sqrt{2}} +\delta, \quad
        \braket{\psi_0|\theta_{\pm}} = \pm\frac{i}{\sqrt{2}} +\delta, \\
        \theta &= \frac{2m}{\sqrt{r_1 r_2}}(1+\delta),   
    \end{align}
    where $\delta = O(\frac{m}{r_2} +\frac{m}{r_1} +\frac{1}{m})$.
    Consider $p(t_2) = \left| \bra{t} (W^{t_1} C_4)^{t_2} \ket{\psi_0} \right|$.
    Let $\Pi_{\pm} :=\ket{\theta_+}\bra{\theta_+} + \ket{\theta_-}\bra{\theta_-}$, and $\Pi_j$ be the projection onto the other $7$ eigenvectors of $W^{t_1}C_4$. Then $\| \Pi_{\pm} \ket{\psi_0}\| = 1-\delta$, and $\left| \bra{t} \Pi_j \ket{\psi_0} \right| \leq \| \Pi_{j} \ket{\psi_0}\| = \delta$.
    Therefore,
    \begin{align}
        p(t_2) &\geq \left| e^{it_2\theta} \braket{t|\theta_+} \braket{\theta_+|\psi_0} + e^{-it_2\theta} \braket{t|\theta_-} \braket{\theta_-|\psi_0} \right| - \sum_{j} \left| \bra{t} \Pi_j \ket{\psi_0} \right| \\
        &\geq \left|e^{it_2\theta} \frac{-i}{2} (1+\delta) + e^{-it_2\theta} \frac{i}{2} (1+\delta) \right| -7\delta \\
        &= \sin(t_2 \theta) -O(\delta).
    \end{align}
    Setting $t_2 = \lfloor \frac{\pi}{4} \frac{\sqrt{r_1 r_2}}{m} \rceil$ we have $p(t_2) = 1 -\delta$ which completes the proof.
\end{proof}

\subsection{Layer 2}\label{subsec:layer_2}
Suppose the input state $\ket{R_1} \ket{D(R_1)}$ (see Eq.~\eqref{eq:overview_D_R_1} for $D(R_1)$) of this layer satisfies $|R_1 \cap \triangle| = 1$, the target state of quantum edge-walk search on $G_2 = J(n_1,r_2)$ is an equal superposition of $\ket{R_2,y}$ such that $|R_2 \cap \triangle| = 1,y\notin\triangle$.

\textbf{Setup.}
The initial state is
\begin{equation}
      \ket{\psi_0} = \ket{R_1} \ket{D(R_1)}
      \frac{1}{\sqrt{n_2}} \sum_{y} \ket{y},
\end{equation}
where $\ket{D(R_1)} = \frac{1}{\sqrt{\binom{n_1}{r_2}}} \sum_{R_2\subseteq[n]-R_1} \ket{R_2} \ket{G_{R_1,R_2}}$.
We only need to construct an equal superposition of $y\in[n]-R_1-R_2$ based on $\ket{R_1}\ket{R_2}$, which requires no query.

\textbf{Update.} A step of quantum walk $U$ consists of $2$ query operations and $2$ diffusion operations:
\begin{equation}
    U = Q \cdot U_B(\theta_2) \cdot Q \cdot U_A(\theta_1),
\end{equation}
with working registers:
\begin{equation}
    \ket{R_1} \ket{R_2} \ket{G_{R_1,R_2}} \ket{y} \ket{G_{R_1,y}}.
\end{equation}
The first diffusion operation $U_A(\theta_1)$ acts on registers $\ket{R_1} \ket{R_2} \ket{y}$, and can be seen as choosing a random $y\in [n] -R_1 -R_2$ to be moved in to $R_2$:
\begin{align}
U_A(\theta_1) &= \sum_{R_1\cap R_2=\phi} \ket{R_1,R_2}\bra{R_1,R_2} \otimes C_{R_1,R_2}(\theta_1), \\
C_{R_1,R_2}(\theta_1) &= I - (1-e^{i\theta_1})\ket{\varphi(R_1,R_2)} \bra{\varphi(R_1,R_2)}, \\
\ket{\varphi(R_1,R_2)} &= \frac{1}{\sqrt{n_2}}\sum_{y\in[n]-R_1-R_2} \ket{y}
\end{align}
The query operation $Q$ calls the oracle $O$ (Eq.~\eqref{eq:oracle_triangle_sum}) for $r_1$ times to update the data $G_{R_1,y}$ associated with $\ket{R_1} \ket{y}$.
The second diffusion operation $U_B(\theta_2)$ acts on all registers, and can be seen as choosing a random $y'\in R_2$ being removed from $R_2$ and at the same time moving $y$ into $R_2$, while updating the associated data $\ket{G_{R_1,R_2}}\ket{G_{R_1,y}}$ simultaneously:
\begin{align}
U_B(\theta_2) &= \sum_{R_1} \ket{R_1}\bra{R_1} \otimes C_{R_1}(\theta_2), \\
C_{R_1}(\theta_2) &= I - (1-e^{i\theta_2})
\sum_{R_2+y\subseteq [n]-R_1}
\sum_{\bar{G} \in [M]^{r_1\times(r_2+1)}}
\ket{\varphi_{R_2+y}^{\bar{G}}} \bra{\varphi_{R_2+y}^{\bar{G}}}, \\
\ket{\varphi_{R_2+y}^{\bar{G}}} &= \frac{1}{\sqrt{r_2+1}} \sum_{y'\in R_2+y} \ket{R_2'} \ket{y'} \ket{G_{R_1,R_2'}} \ket{G_{R_1,y'}}.
\end{align}
In $\ket{\varphi_{R_2+y}^{\bar{G}}}$, suppose the horizontal juxtaposition $[G_{R_1,R_2}, G_{R_1,y}]$ equals $\bar{G}$, then for $\ket{R_2'} \ket{y'}$ s.t. $R_2' = R_2 - y' +y$, its associated data $[G_{R_1,R_2'}, G_{R_1,y'}]$ is an appropriate permutation of the columns of $\bar{G}$.
This is similar to $\ket{\varphi_{S_i + z_i}^{\bar{G}_i}}$ defined in Eq.~\eqref{eq:layer_4_phi_B}.

\textbf{Checking.} The checking operator $S_\mathcal{M}(\alpha) = C_2 (I\otimes \mathrm{diag}(1,e^{i\alpha})) C_2$ adds phase shift $e^{i\alpha}$ to states $\ket{R_1} \ket{R_2} \ket{G_{R_1,R_2}} \ket{y}$ which satisfies $|R_1 \cap \triangle| = 1$, $|R_2 \cap \triangle| = 1$ and $y\notin\triangle$.

Assume $|R_1 \cap \triangle| = 1$,
then $|([n] -R_1) \cap \triangle| = 2$.
Therefore, the quantum vertex-walk search on $G_2 = J(n_1,r_2)$ satisfies Condition~\ref{cond:five} with $K = \triangle - R_1$ and $(j_0,l_0) = (1,0)$, and we can obtain the same $5$-dimensional invariant subspace as shown by Fig.~\ref{fig:pre_subspace} in Section~\ref{subsec:pre_vertex},
but now $N := n_1$, $r := r_2$, $R := R_2$, and $K := \triangle - R_1$.
Since the target state is now $\ket{t} = \ket{1,0}$, the proportion of marked states becomes
\begin{align}
    \epsilon_2 &= \frac{2 \binom{n_1-2}{r_2-1} (n_1-r_2-1)}{\binom{n_1}{r_2}(n_1-r_2)} \\
    &= 2\frac{r_2}{n_1}(1-\frac{r_2}{n_1-1})
    =\Theta(\frac{r_2}{n}).
\end{align}
Thus by combing Lemma~\ref{lem:eedp} (where now $t \in O(\sqrt{r_2})$) and Lemma~\ref{lem:beta_fixed} (let $S_{\psi_0}(-\beta) = U^{t}$ and $\lambda = \epsilon_2$),
we can obtain $\ket{t}$ with certainty, and the query complexity is $0 + \sqrt{{n}/{r_2}} (\sqrt{r_2} r_1 + c_2)$, same as Eq.~\eqref{eq:overview_C_1_2nd}.

\subsection{Layer 1}\label{subsec:layer_1}
The goal of this layer is to construct an equal superposition of $\ket{R_1}\ket{R_1'}$ such that $|R_1\cap\triangle|=1$ and $|R_1'\cap\triangle|=1$.
The quantum edge-walk search on $G_1 = J(n,r_1)$ is the same as in Section~\ref{subsec:pre_edge}.
We first describe how to implement the update operator in the following lemma.

\begin{lemma}\label{lem:layer_1_update}
    The `Data' operator that transform the data $D(R_1)$ to $D(R_1')$ based on $\ket{R_1'}\ket{R_1}$ requires $2(r_1 +r_2)$ queries.
\end{lemma}
\begin{proof}
    The transformation consists of the following two steps:
    \begin{align}
        \ket{D(R_1)} =& \frac{1}{\sqrt{\binom{n_1}{r_2}}} \sum_{R_2 \cap R_1 = \phi} \ket{R_2} \ket{G_{R_1,R_2}} \nonumber\\
        U_{1,1} \mapsto & \frac{1}{\sqrt{\binom{n_1}{r_2}}} \sum_{R_2' \cap R_1' = \phi} \ket{R_2'} \ket{G_{R_1,R_2'}} \\
        U_{1,2} \mapsto & \frac{1}{\sqrt{\binom{n_1}{r_2}}} \sum_{R_2' \cap R_1' = \phi} \ket{R_2'} \ket{G_{R_1',R_2'}} = \ket{D(R_1')}.
    \end{align}
    Specifically, let $x_1 := R_1\setminus R_1'$ and $x_1' = R_1'\setminus R_1$, which can be calculated from $\ket{R_1} \ket{R_1'}$.
    Consider some $R_2 \subseteq [n] -R_1$.
    If $x_1' \notin R_2$, then $U_{1,1}$ keeps $R_2$ unchanged;
    if $x_1' \in R_2$, then $U_{1,1}$ implements $R_2 \mapsto R_2' =  R_2 - x_1' + x_1$ and also update the data $\ket{G_{R_1,R_2}}$ simultaneously: it clears the data $G_{R_1,x_1'}$ in column $x_1'$, and then writes back $G_{R_1,x_1}$, and finally permutes the $r_2$ columns so that the indexes $R_2' =  R_2 - x_1' + x_1$ remain the ascending order.
    Therefore, $U_{1,1}$ requires $2r_1$ queries.
    
    Operator $U_{1,2}$ does not need to deal with separate cases, and is simpler to implement.
    It updates $G_{x_1,R_2'}$ to $G_{x_1',R_2'}$ in row $x_1$, and then permutes the $r_1$ rows so that the indexes $R_1' = R_1 -x_1 +x_1'$ remain the ascending order.
    Therefore, $U_{1,2}$ requires $2r_2$ queries.
    The above process is illustrated in Fig.~\ref{fig:layer_1_update}.
    \begin{figure}
        \centering
        \includegraphics[width=0.6\linewidth]{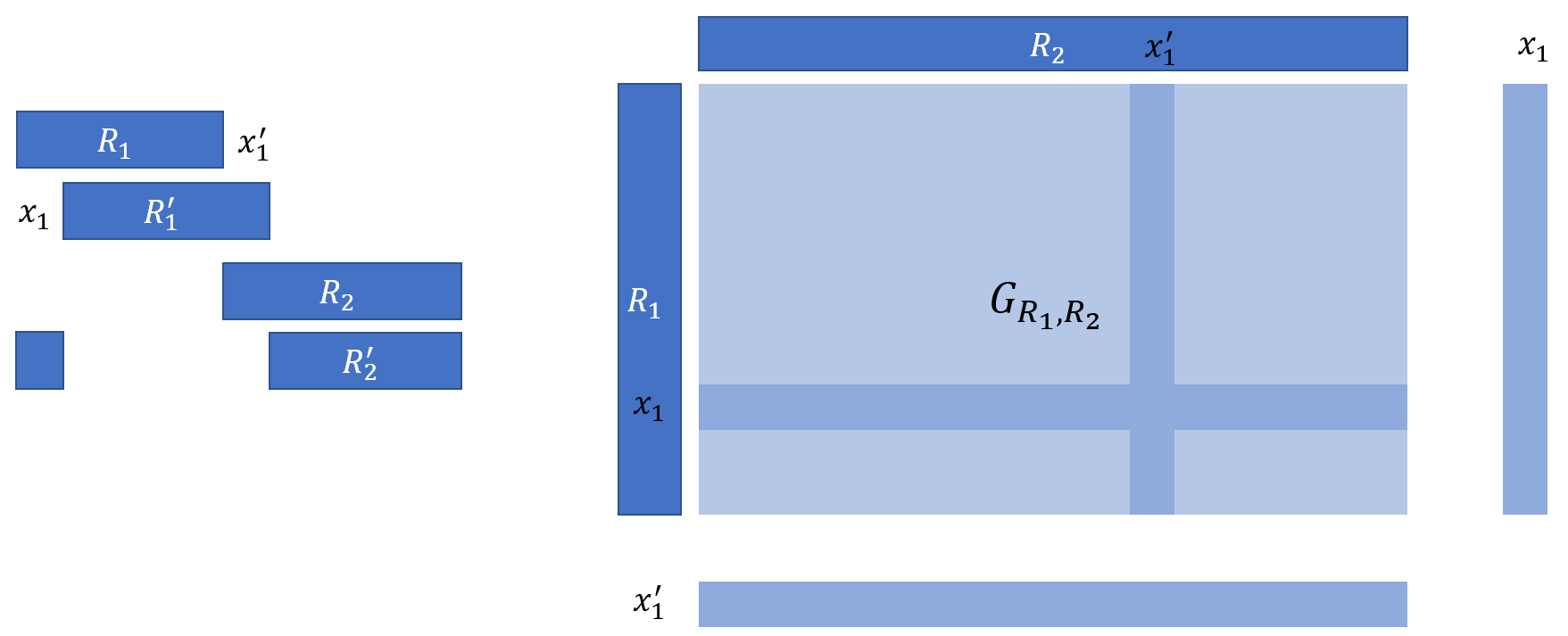}
        \caption{The update operator in layer 1 requires $2r_1 + 2r_2$ queries as shown in Lemma~\ref{lem:layer_1_update}.}
        \label{fig:layer_1_update}
    \end{figure}
\end{proof}

\textbf{Checking.} The checking operator $C_1$ adds phase shift to $\ket{R_1}\ket{R_1'} \ket{D(R_1)}$ such that $|R_1\cap\triangle|=1$ and $|R_1'\cap\triangle|=1$.
Suppose $\bar{C}_1$ implemented in layer 2 flips an auxiliary qubit based on $\ket{D(R_1)}$ if $|R_1\cap\triangle|=1$, then $C_1$ can be implemented with costs $c_1 = 2u_1 + 4\bar{c}_1$ as shown below:

1. apply $\bar{C}_1$ to $\ket{R_1}\ket{D(R_1)}$ and store the result on $\ket{b_1}$ initialized to $\ket{0}$; 
2. update the data $\ket{D(R_1)}$ to $\ket{D(R_1')}$ using Lemma~\ref{lem:layer_1_update}; 
3. apply $\bar{C}_1$ to $\ket{R_1'}\ket{D(R_1')}$ and store the result on $\ket{b_1'}$ initialized to $\ket{0}$;
4. use phase kick-back to add phase shift $(-1)$ if $b_1 \wedge b_1' = 1$;
5. clear the register $\ket{b_1}$, $\ket{b_1'}$ and recover $D(R_1)$ by applying the inverse of the first 3 steps.

\textbf{Invariant subspace.} Denote by $\ket{j,l}$ the equal superposition of states $\ket{R_1}\ket{R_1'}$ satisfying $\left| R_1 \cap \triangle \right| = j$ and $\left| R_1' \cap \triangle  \right| = l$.
The $10$ basis states of the quantum walk's invariant subspace $\mathcal{H}_0$ is illustrated in Fig.~\ref{fig:layer_1_subspace}.

\begin{figure}
    \centering
    \includegraphics[width=0.5\textwidth]{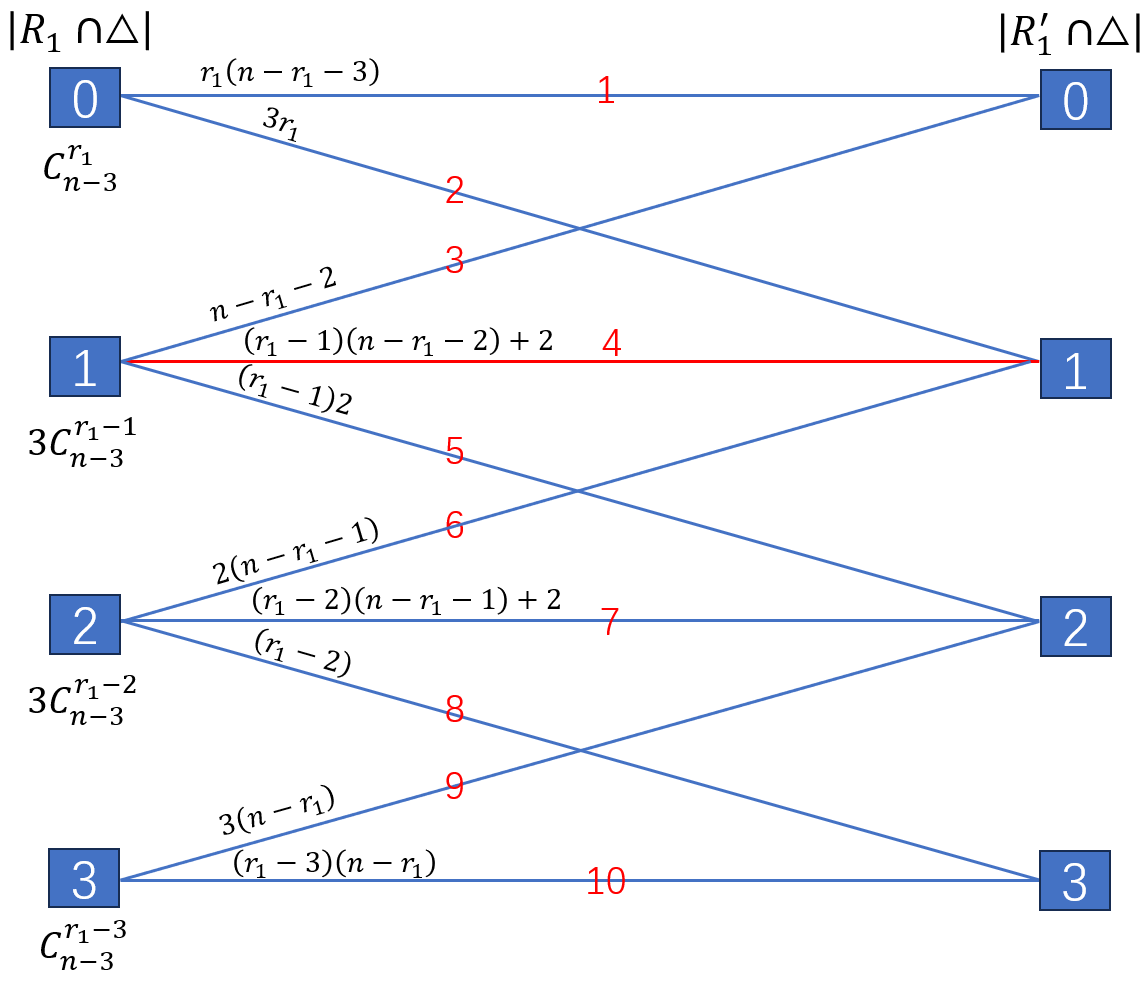}
    \caption{Illustration of the $10$ basis states $\ket{j,l}$ of the invariant subspace $\mathcal{H}_0$ of quantum walk search on $G_1 = J(n,r_1)$. }
    \label{fig:layer_1_subspace}
\end{figure}

It can be seen from Fig.~\ref{fig:layer_1_subspace} and the definition of a step of quantum walk $U$ which consists of `Coin' and `Swap' operators as shown in Section~\ref{subsec:pre_edge}, that $U$ takes the following matrix form $W$ in $\mathcal{H}_0$:
\begin{equation}
    W = S(2AA^\dagger-I),
\end{equation}
where $S = \mathrm{diag}(1,X,1,X,1,X,1)$, and non-negative matrix $A$ satisfies
\begin{equation}
A.^2 = \begin{bmatrix}
1 - \frac{3}{n-r} & 0 & 0 & 0 \\
\frac{3}{n-r} & 0 & 0 & 0 \\
0 & \frac{n-r-2}{r(n-r)} & 0 & 0 \\
0 & \frac{(r-1)(n-r-2)+2}{r(n-r)} & 0 & 0 \\
0 & \frac{2(r-1)}{r(n-r)} & 0 & 0 \\
0 & 0 & \frac{2(n-r-1)}{r(n-r)} & 0 \\
0 & 0 & \frac{(r-2)(n-r-1)+2}{r(n-r)} & 0 \\
0 & 0 & \frac{r-2}{r(n-r)} & 0 \\
0 & 0 & 0 & \frac{3}{r} \\
0 & 0 & 0 & 1-\frac{3}{r}
\end{bmatrix}.
\end{equation}
The initial state takes the following form in $\mathcal{H}_0$:
\begin{equation}
\ket{\psi_0}.^2 =
\begin{bmatrix}
(n-r_1-1)(n-r_1-2)(n-r_1-3) \\
3(n-r_1-1)(n-r_1-2) \\
3(n-r_1-1)(n-r_1-2) \\
3(n-r_1-1)((r_1-1)(n-r_1-2)+2) \\
6(r_1-1)(n-r_1-1) \\
6(r_1-1)(n-r_1-1) \\
3(r_1-1)((r_1-2)(n-r_1-1)+2) \\
3(r_1-1)(r_1-2) \\
3(r_1-1)(r_1-2) \\
(r_1-1)(r_1-2)(r_1-3)
\end{bmatrix}
\div (n(n-1)(n-2)).
\end{equation}
It can be seen that the overlap between $\ket{\psi_0}$ and the target state $\ket{t} = \ket{e_4}$ is $\sqrt{\epsilon_1} = \Theta(\sqrt{r_1/n})$.

\begin{lemma}\label{lem:layer_1_success}
    Setting $t_1 = \lfloor \frac{\pi}{2} \sqrt{2 r_1} \rceil$ and $t_2 = \lfloor \frac{\pi}{4} \sqrt{\frac{n}{3 r_1}} \rceil$, the success amplitude $p = \left| \bra{t} (W^{t_1} C_1)^{t_2} \ket{\psi_0} \right| $ satisfies $ p = 1 - O(\frac{1}{r_1} +\frac{r_1}{n})$.
\end{lemma}
\begin{proof}
    Similar to \cite[Lemma 2]{subset}, it can be shown that $W$ has two eigenvectors $\ket{u_\pm}$ with eigenvalues $e^{\pm i \varphi}$ such that
    \begin{align}
        \left| \braket{t|u_{\pm}} \right| &= \frac{1}{\sqrt{2}} +\delta, \\
        \varphi &= \sqrt{\frac{2}{r}} (1+\delta),
    \end{align}
    where $\delta := O(\frac{1}{r_1} +\frac{r_1}{n})$.
    To make $t_1 \varphi \approx \pi$, we set $t_1$ to be the nearest integer to $\frac{\pi}{2} \sqrt{2 r_1}$.
    Note that $C_1 = 2\ket{t}\bra{t} - I$ adds phase shift to the only target basis state $\ket{t} = \ket{e_4}$ in $\mathcal{H}_0$.
    Thus, similar to \cite[Lemma 3]{subset}, it can be shown that $W^{t_1}C_1$ has two eigenvectors $\ket{\theta_\pm}$ with eigenvalues $e^{\pm i\theta}$ such that
    \begin{align}
        \braket{t|\theta_{\pm}} &= \frac{1}{\sqrt{2}} +\delta, \quad
        \braket{\psi_0|\theta_{\pm}} = \pm\frac{i}{\sqrt{2}} +\delta, \\
        \theta &= 2 \sqrt{\frac{3r}{n}}(1+\delta).        
    \end{align}
    Consider $p(t_2) = \left| \bra{t} (W^{t_1} C_1)^{t_2} \ket{\psi_0} \right|$.
    Let $\Pi_{\pm} :=\ket{\theta_+}\bra{\theta_+} + \ket{\theta_-}\bra{\theta_-}$,
    and $\Pi_j$ be the projection onto the other $8$ eigenvectors of $W^{t_1}C_1$.
    Then $\| \Pi_{\pm} \ket{\psi_0}\| = 1-\delta$,
    and $\left| \bra{t} \Pi_j \ket{\psi_0} \right| \leq \| \Pi_{j} \ket{\psi_0}\| = \delta$.
    Therefore,
    \begin{align}
        p(t_2) &\geq \left| e^{it_2\theta} \braket{t|\theta_+} \braket{\theta_+|\psi_0} + e^{-it_2\theta} \braket{t|\theta_-} \braket{\theta_-|\psi_0} \right| - \sum_{j} \left| \bra{t} \Pi_j \ket{\psi_0} \right| \\
        &\geq \left|e^{it_2\theta} \frac{-i}{2} (1+\delta) + e^{-it_2\theta} \frac{i}{2} (1+\delta) \right| -8\delta \\
        &= \sin(t_2 \theta) -O(\delta).
    \end{align}
    Setting $t_2 = \lfloor \frac{\pi}{4} \sqrt{\frac{n}{3 r_1}} \rceil$,
    we have $p(t_2) = 1 -\delta$,
    which completes the proof.
\end{proof}

\section{Discussions}\label{sec:conclude}
In this paper, we have shown that there is a deterministic quantum algorithm for the triangle sum promised problem (i.e. the graph contains at most one target triangle) based on derandomization of a nested-quantum-walk-based algorithm by Jeffery et al.
Our algorithm achieves the same $O(n^{9/7})$ queries with the state-of-the-art bounded error quantum algorithm, utilizing several non-trivial techniques.
It may be worth  further considering the following problems.
\begin{enumerate}
    \item Will the lower bound of $\Omega(n^{9/7}/\sqrt{\log n})$ remains unchanged for the triangle sum promised problem considered in this paper?

    \item Is it still possible to design a deterministic quantum algorithm when the graph is promised to contain none or $k>1$ target triangles?

    \item Is it possible to derandomize the state-of-the-art $O(n^{5/4})$-query quantum algorithm for triangle finding~\cite{triangle_extended} when given an additional promise?
\end{enumerate}

\bibliographystyle{quantum}
\bibliography{ref}

\end{document}